\documentclass[10pt,conference]{IEEEtran}
\usepackage{amsfonts}
\usepackage{amssymb}
\usepackage{amsmath}
\usepackage{algorithm}
\usepackage{algpseudocode}
\usepackage{algorithmicx}

\newcommand{\comment}[1]{}

\newtheorem{theorem}{Theorem}[section]
     \newtheorem{lemma}[theorem]{Lemma}
     
     \newtheorem{corollary}[theorem]{Corollary}

     \newcommand{\qed}{\nobreak \ifvmode \relax \else
           \ifdim\lastskip<1.5em \hskip-\lastskip
           \hskip1.5em plus0em minus0.5em \fi \nobreak
           \vrule height0.75em width0.5em depth0.25em\fi}
\newcounter{remark_ordering}
\newtheorem{remark}[remark_ordering]{Remark}

\begin{document}

% paper title
\title{Performance Analysis for Data Compression Based Signal Classification Methods}

\author{
\authorblockN{Xudong Ma}
\authorblockA{Pattern Technology Lab LLC, Delaware, U.S.A.  \\
Email: xma@ieee.org}}

\maketitle

\comment{
\begin{abstract}
This paper provides the instructions for the preparation of papers
for submission to ISIT 2007 and relevant style file that produced
this page.
\end{abstract}

\section{Introduction}
The 2007 IEEE International Symposium on Information Theory will be held
at the Acropolis Congress and Exhibition Center in Nice, France,
from Sunday June 24 through Friday June 29, 2007.

\section{Submission and Review Process}
Papers will be reviewed on the basis of a manuscript of
sufficient detail to permit reasonable evaluation.
The manuscript should
{\bf not exceed five double-column pages,
with single line spacing, main text
font size no smaller than 10 points,
and at least 3/4 inch margins (about 18 mm)}.
The deadline for submission is {\bf January 8, 2007}, with
notification of decisions by {\bf March 25, 2007}.

The deadline and five page limit will be strictly
enforced. In view of the large number of submissions expected,
multiple submissions by the same author will receive especially
stringent scrutiny. All accepted papers will be allowed twenty
minutes for presentation.

\section{Proceedings}
Accepted papers will be published in full (up to
five pages in length) on CD-ROM. A hard-copy book of abstracts will
also be distributed at the Symposium to serve as a guide to the
sessions.

The deadline for the submission of the final camera-ready
paper is {\bf April 20, 2007}. Final manuscript guidelines will be
made available after the notification of decisions.

\section{Preparation of the Paper}
Only electronic submissions in form of a Postscript (PS) or Portable Document
Format (PDF) file will
be accepted. Most authors will prepare their papers with \LaTeX. The
\LaTeX\ style file (\verb#IEEEtran.cls#) and the \LaTeX\ source
(\verb#isit_example.tex#) that produced this page may be downloaded
from the ISIT 2007 web site (http://www.isit2007.org).
Do not change the style file in any way.
Authors using other means to prepare their manuscripts should attempt
to duplicate the style of this example as closely as possible.

% The {\bf Abstract} section should be no more than $250$ words and
% should contain no math notation. If citations are required in the abstract,
% they should be self-contained, e.g. Shannon, \emph{Bell Syst.\ Tech.\ J.} 1948,
% rather than [1]. The abstract will be published separately in the hardcopy book
% of abstracts.

The style of references, e.g.,
\cite{Shannon1948}, equations, figures, tables, etc., should be the
same as for the \emph{IEEE Transactions on Information Theory}. The
affiliation shown for authors should constitute a sufficient mailing
address for persons who wish to write for more details about the
paper.

\section{Electronic Submission}
The paper submission portal is EDAS:
\begin{verbatim}
http://edas.info/newPaper.php?c=5045&
\end{verbatim}
Check ISIT2007 website (http://www.isit2007.org) for relevant
announcements.

\section{Conclusion}
Never conclude a real information theoretic paper.
If you have to, the conclusion goes here.

% conference papers do not normally have an appendix

% use section* for acknowledgement
\section*{Acknowledgment}
% optional entry into table of contents (if used)
%\addcontentsline{toc}{section}{Acknowledgment}
The authors would like to thank various sponsors for supporting their research.
In particular, we thank the TPC chairs of ISIT 2006 for
providing the \LaTeX\ templates for paper submission.

% trigger a \newpage just before the given reference
% number - used to balance the columns on the last page
% adjust value as needed - may need to be readjusted if
% the document is modified later
%\IEEEtriggeratref{8}
% The "triggered" command can be changed if desired:
%\IEEEtriggercmd{\enlargethispage{-5in}}

% references section
% NOTE: BibTeX documentation can be easily obtained at:
% http://www.ctan.org/tex-archive/biblio/bibtex/contrib/doc/

% can use a bibliography generated by BibTeX as a .bbl file
% standard IEEE bibliography style from:
% http://www.ctan.org/tex-archive/macros/latex/contrib/supported/IEEEtran/bibtex
%\bibliographystyle{IEEEtran.bst}
% argument is your BibTeX string definitions and bibliography database(s)
%\bibliography{IEEEabrv,../bib/paper}
%
% <OR> manually copy in the resultant .bbl file
% set second argument of \begin to the number of references
% (used to reserve space for the reference number labels box)
}

\begin{abstract}

In this paper, we present an information theoretic analysis of the
blind signal classification algorithm. We show that the algorithm is
equivalent to a Maximum A Posteriori (MAP) estimator based on
estimated parametric probability models. We prove a lower bound on
the error exponents of the parametric model estimation. It is shown 
that the estimated model parameters converge in probability to the
true model parameters except some small bias terms.

\end{abstract}

\section{Introduction}

In this paper, we consider the blind signal classification problems.
In the considered scenarios, a sequence of random signal samples
$X_1$, $X_2$, \ldots, $X_N$ is observed, where each signal sample is
a real number or a vector in a finite dimensional space. It is
assumed that the $N$ signal samples are generated by $J$ information
sources with different statistical properties. However, it is
unknown from which information source each signal sample $X_n$ is
emitted. The blind signal classification problems denote the
problems of estimating the membership of each signal sample to the
information sources. The signal classification problems find
applications in many areas of image processing, computer vision and
machine learning, for example, in image segmentation, and cluster
analysis. For background information in these applications, we refer
interested readers to \cite{sonka} \cite{xu} and references therein.

In \cite{ma09}, a novel algorithm for the signal classification
problems is proposed based on data compression. The algorithm is
based on the intuitive idea that optimal classification induces
optimal adaptive data compression. Therefore, the signal
classification problems can be formulated as optimization problems.
An analysis in an algorithmic viewpoint was also presented. It was
shown in the paper \cite{ma09} that a soft membership relaxation can
be used to reduce the computational complexity with asymptotic
vanishing optimality loss. Simulation results show that the
algorithm has nice performance.

It is well known that there exist close connections between
information theory and statistical inference. Especially, source
coding and data compression have been used in statistical inference
problems, such as prediction, estimation and modeling, see for
instance \cite{rissanen84, rissanen86, rissanen96, grunwald}.
However, there exists no discussion on using data compression for
classification and clustering until very recent. In
\cite{cilibrasi05, cilibrasi04}, a ``clustering by compression''
algorithm has been proposed. The approach in \cite{cilibrasi05,
cilibrasi04} is different from the approach in \cite{ma09} in terms
of their ways of using data compression. In \cite{cilibrasi05,
cilibrasi04}, the data compression methods are used to compute
distances between data items. The clustering results are then
obtained by using conventional methods based on the computed
distances.

In this paper, we present an information theoretical analysis to
justify the intuitive idea of the blind signal classification
algorithm in \cite{ma09}. It is shown that the blind signal
classification algorithm is equivalent to a Maximum A Posteriori
(MAP) estimator based on estimated parametric probability models. We
also discuss the error exponents of the model parameter estimation.
It is shown that the estimated model parameters converge to the true
model parameters in probability. These theoretical discussions
suggest that the blind signal classification algorithm has nice
performance.

The discussions in this paper focus on the cases that the
information sources are independent and identically distributed
(i.i.d.) Gaussian, and there are two information sources. Even
though, more sophisticated cases are not covered in this paper, the
discussions presented here can provide useful insights into these
more general cases. The discussions can be easily generalized to the
cases of multiple information sources. With additional works, the
results can also be generalized to the cases of non-Gaussian,
Markov, stationary or ergodic information sources.

Notation: we use $\lfloor \cdot \rfloor$ to denote the floor
function, that is, $\lfloor x \rfloor$ is the largest integer
smaller than $x$. We use $\ln(\cdot)$ to denote the logarithmic
function with base $e$. We use $H(P)$ and $D(P||Q)$ to denote the
entropy function and information divergence respectively. If $P$ and
$Q$ are discrete probability mass functions, then
\begin{align}
& H(P)=\sum_{x\in\mathcal{X}}-P(x)\ln(P(x)),\notag \\
& D(P||Q)=\sum_{x\in\mathcal{X}}P(x)\ln \frac{P(x)}{Q(x)},
\end{align}
where $\mathcal{X}$ is the discrete alphabet. If $P$ and $Q$ are
probability density functions, then
\begin{align}
& H(P)=\int_{-\infty}^{\infty}-P(x)\ln(P(x))dx,\notag \\ &
D(P||Q)=\int_{-\infty}^{\infty}P(x)\ln \frac{P(x)}{Q(x)}dx.
\end{align}
If $f(N)$ and $g(N)$ are two functions of the number $N$, we use
$f(N)\sim g(N)$ to denote that
\begin{align}
\lim_{N\rightarrow \infty}
\frac{1}{N}\ln\left(\frac{f(N)}{g(N)}\right)=0.
\end{align}
For a sequence $x_1x_2\ldots x_N$ over a discrete alphabet $\mathcal
X$, the type of the sequence is defined as the corresponding
empirical distribution $P$ over $\mathcal{X}$, that is, $P(a)$ is
equal to the fraction of $x_i$ taking value $a$. For a type $P$, the
type class $\mathcal{T}(P,N)$ is the set of sequences with length
$N$ and type  $P$. We use $\mathbb{P}(\mathcal{T}(P,N))$ to denote
the probability that the type of the random sequence $x_1x_2\ldots
x_N$ is $P$.

The rest of this paper is organized as follows. We review the data
compression based signal classification algorithm in Section
\ref{section_classification_algorithm}. We discuss the necessary
conditions for the optimal solutions of the blind signal
classification algorithm in Section \ref{sec_optimal_condition}. We
discuss the error exponents of the parameter estimation in Section
\ref{sec_single_letter_bound}. The concluding remarks are presented
in Section \ref{sec_conclusion}.

%%%%%%%%%%%%%%%%%%%%%%%%%%%%%%%%%%%%%%%%%%%%%%%%%%%%%%%%%%%%%%%%%%%%%%%%%%

%

\section{Blind Signal Classification Algorithm}
\label{section_classification_algorithm}

In this paper, we consider the scenario, where a sequence of random
real-valued signal samples $X_1,X_2,\ldots,X_N$ is observed. Each
signal sample $X_n$ is independently drawn from one of the several
i.i.d. Gaussian information sources. The probability density
function is
\begin{align}
P(X_n) =
\sum_{i=1}^{J}\frac{\alpha_i}{\sqrt{2\pi}\sigma_i}\exp\left(\frac{-(x_n-\mu_i)^2}{2\sigma_i^2}\right),
\end{align}
where $J$ is the total number of information sources, $\alpha_i$ is
the probability that $X_n$ is drawn from the $i$-th information
source, and $\mu_i$, $\sigma_i$ are the Gaussian distribution
parameters of the $i$-th information source. The goal is to estimate
the membership of each signal sample $x_n$ to the $J$ information
sources.

The blind signal classification  algorithm in \cite{ma09} is based
on a data compression argument that accurate signal classification
results in accurate signal modeling and efficient data compression.
Therefore, the signal samples should be classified, so that the
coding efficiency is maximized. Let $m_{ni}$ denote the membership
variable for the $n$-th signal sample with respect to the $i$-th
class,
\begin{align}
m_{ni}=\left\{\begin{array}{ll} 1, & \mbox{if }x_n\mbox{ is
classified
into the }i\mbox{-th class} \\
0, & \mbox{otherwise}
\end{array}\right.
\end{align}
The algorithm searches for the optimal values $m_{ni}$ such that the
following objective function is minimized,
\begin{align}
\label{algorithm_objective_function}
G=\exp\left(2H(\alpha_1,\ldots,\alpha_J)\right)\prod_{i=1}^{J}\left(\sigma_i^2\right)^{\alpha_i}
\end{align}
where $\alpha_i$ is the fraction of signal samples that are
classified into the $i$-th class,  $\sigma_i^2$ is the variance of
signal samples in the $i$-th class, $H(\alpha_1,\ldots,\alpha_j)$ is
the entropy function in nats,
\begin{align}
& \alpha_i= \frac{\sum_n m_{ni}}{N},\,\,\, \mu_i= \frac{\sum_n
m_{ni}x_n}{\alpha_iN},\,\,\, \sigma_i^2= \frac{\sum_n
m_{ni}(x_n-\mu_i)^2}{\alpha_iN}.
\end{align}
The objective function $G$ relates to the so called {\it
classification gain} and adaptive coding efficiency \cite{ma09}.

In \cite{ma09}, the {\it hard} membership variables $m_{ni}$ are
relaxed into {\it soft} membership variables. That is, $m_{ni}$ can
take real values, such that $ \sum_i m_{ni}=1$, $0\leq m_{ni}\leq
1$. It is proven in \cite{ma09} that the optimality loss due to the
relaxation vanishes asymptotically.

\section{Necessary Condition for Optimization Solutions}

\label{sec_optimal_condition}

In this section, we show a necessary condition for the solution in
the above blind signal classification method with soft membership
variables. It turns out that useful insights can be gained from the
necessary condition.

Let us assume that the probability density function of $X_n$ is
$P^\ast (X_n)  = \alpha_1^\ast f_1^\ast(x) + \alpha_2^\ast
f_2^\ast(x) $, where
\begin{align}
f_i^\ast(x)  =
\frac{1}{\sqrt{2\pi}\sigma_i^\ast}\exp\left(\frac{-(x-\mu_i^\ast)^2}{2(\sigma_i^\ast)^2}\right).
\end{align}
Let $\widehat{m}_{ni}$ denote one global minimizer of the
optimization programming in the blind signal classification method.
Let $\widehat{\alpha}_1$, $\widehat{\alpha}_2$, $\widehat{\mu}_1$,
$\widehat{\mu}_2$, $\widehat{\sigma}_1$, $\widehat{\sigma}_2$ denote
the distribution parameters corresponding to the optimization
solution $\widehat{m}_{ni}$.

\begin{theorem}
\label{theorem_optimal_char_one} The optimal solution
$\widehat{m}_{ni}$ satisfies the following condition.
\begin{align}
\widehat{m}_{n1} = \left\{\begin{array}{ll} 1, & \mbox{ if }
\widehat{\alpha}_1\widehat{ f}_1(x_n)>\widehat{\alpha}_2
\widehat{ f}_2(x_n) \\
0, & \mbox{ if }\widehat{\alpha}_1\widehat{ f}_1(x_n)<
\widehat{\alpha}_2\widehat{ f}_2(x_n)
\end{array}\right.
\end{align}
where $\widehat{f}_1$ and $\widehat{f}_2$ are the Gaussian
probability density functions corresponding to the parameters
$\widehat{\mu}_1,\widehat{\mu}_2,\widehat{\sigma}_1,\widehat{\sigma}_2$,
\begin{align}
\widehat{f}_i (x) = \frac{1}{\sqrt{2\pi}\widehat{\sigma}_i}
\exp\left(\frac{-(x-\widehat{\mu}_i)^2}{2\widehat{\sigma}_i^2}\right).
\end{align}
\end{theorem}
\begin{proof}\emph{(sketch)}
Consider $G$ as a function solely determined by $m_{ni}$. Taking a
derivative, we have
\begin{align}
 \frac{\partial (\ln(G))}{\partial m_{n1}}  = &
 \frac{1}{N}\left[ \ln\left( \widehat{\sigma}_1^2 \right)
-2\ln(\widehat{\alpha}_1)
 + \frac{(x_n-\widehat{\mu}_1)^2}{\widehat{\sigma}_1^2}\right]
\notag \\ & - \frac{1}{N}\left[ \ln\left( \widehat{\sigma}_2^2
\right)-2\ln(\widehat{\alpha}_2)  +
\frac{(x_n-\widehat{\mu}_2)^2}{\widehat{\sigma}_2^2}\right] \notag
\\
  = &
 \frac{2}{N}
\ln \left[\sqrt{2\pi} \widehat{\alpha}_2\widehat{f}_2(x_n) \right] -
\frac{2}{N} \ln \left[\sqrt{2\pi}
\widehat{\alpha}_1\widehat{f}_1(x_n) \right]
\end{align}
Therefore,
\begin{align}
& \frac{\partial \ln G}{\partial m_{n1}} <0,\,\,\,\mbox{ if
}\widehat{\alpha}_1\widehat{ f}_1(x_n)> \widehat{\alpha}_2\widehat{
f}_2(x_n) \\
& \frac{\partial \ln G}{\partial m_{n1}}
>0,\,\,\,\mbox{ if }\widehat{\alpha}_1\widehat{ f}_1(x_n)<
\widehat{\alpha}_2\widehat{ f}_2(x_n)
\end{align}
The theorem then follows from the KKT condition \cite{cong01}.
\end{proof}

\begin{corollary}
\label{theorem_optimal_char_two} Define three subsets $\mathcal{A},
\mathcal{B}, \mathcal{C}$ of real numbers as follows, ${\mathcal
A}=\left\{x|\widehat{\alpha}_1\widehat{f}_1(x)>\widehat{\alpha}_2\widehat{f}_2(x)\right\},
{\mathcal
B}=\left\{x|\widehat{\alpha}_1\widehat{f}_1(x)=\widehat{\alpha}_2\widehat{f}_2(x)\right\},
{\mathcal
C}=\left\{x|\widehat{\alpha}_1\widehat{f}_1(x)<\widehat{\alpha}_2\widehat{f}_2(x)\right\}$.
We write $n\in \mathcal{A}$ ($n\in\mathcal{B}$, $n\in\mathcal{C}$),
if $x_n\in\mathcal{A}$ ($x_n\in\mathcal{B}$,  $x_n\in\mathcal{C}$).
Then, the following statements hold.
\begin{align}
\label{optimal_cond_5} \widehat{\alpha}_1= \frac{\sum_{n\in
{\mathcal A}} 1 + \sum_{n\in{\mathcal B}}\widehat{m}_{n1}}{N},
\,\,\, \widehat{\alpha}_2= \frac{\sum_{n\in {\mathcal C}} 1 +
\sum_{n\in{\mathcal B}}\widehat{m}_{n2}}{N},
\end{align}
\begin{align}
\label{optimal_cond_7} \widehat{\mu}_1=\frac{\sum_{n\in {\mathcal
A}}x_n+\sum_{n\in{\mathcal
B}}\widehat{m}_{n1}x_n}{\widehat{\alpha}_1 N},
\end{align}
\begin{align}
\widehat{\mu}_2=\frac{\sum_{n\in {\mathcal
C}}x_n+\sum_{n\in{\mathcal
B}}\widehat{m}_{n2}x_n}{\widehat{\alpha}_2 N},
\end{align}
\begin{align}
\label{optimal_cond_9} \widehat{\sigma}_1=\frac{\sum_{n\in {\mathcal
A}} (x_n-\widehat{\mu}_1)^2+ \sum_{n\in{\mathcal
B}}\widehat{m}_{n1}(x_n-\widehat{\mu}_1)^2 }{\widehat{\alpha}_1  N},
\end{align}
\begin{align}
\label{optimal_cond_10} \widehat{\sigma}_2=\frac{\sum_{n\in
{\mathcal C}} (x_n-\widehat{\mu}_2)^2+ \sum_{n\in{\mathcal
B}}\widehat{m}_{n2}(x_n-\widehat{\mu}_2)^2 }{\widehat{\alpha}_2  N}.
\end{align}
\end{corollary}

\begin{remark}
Theorem \ref{theorem_optimal_char_one} shows that the data
compression based signal classification method is equivalent to the
MAP estimation based on the estimated parametric probability models.
Even though the probability model estimation is just a by-product of
the classification algorithm, the accuracy of model estimation is
critical to the performance of the algorithm.
\end{remark}

\section{Error Exponent of Parameter Estimation}

\label{sec_single_letter_bound}

In this section, we investigate the accuracy of the parametric
probability model estimation in the proposed signal classification
method by using the method of types \cite{csiszar98}. We need to
introduce several auxiliary discrete probability distributions. Let
$M_N, L_N,W_N$ be numbers, which only depend on the number of signal
samples $N$,
\begin{align} M_N=c
(N)^{\frac{1}{2}+\zeta},\,\,\,L_N =
2\left\lfloor\frac{N^{1-\eta}}{\log
N}\right\rfloor+1,\,\,\,W_N=\frac{2M_N}{L_N},
\end{align}
where $c,\zeta,\eta$ are positive constants, and $\zeta+\eta<1/2$.
Let $a_k$ denote the number of signal samples, which fall in the
interval $\left[(k-1/2)W_N,(k+1/2)W_N\right]$. Let $\mathcal{O}_N$
denote the random event that $|x_n|\geq M_N$ for some $n$, $1\leq n
\leq N$. If $\mathcal{O}_N$ does not occur, then
$\{\ldots,a_k/N,\ldots \}$ is a well defined empirical probability
mass function, where
\begin{align}\frac{-L_N+1}{2}\leq k\leq
\frac{L_N-1}{2}.\end{align} We write $k\in \mathcal{A}$, if the
interval $\left[(k-1/2)W_N,(k+1/2)W_N\right]\subset\mathcal{A}$. We
write $k\in \mathcal{C}$, if
$\left[(k-1/2)W_N,(k+1/2)W_N\right]\subset\mathcal{C}$. We write
$k\in \mathcal{B}$, if $\left[(k-1/2)W_N,(k+1/2)W_N\right]\cap
\mathcal{B}\neq \emptyset$. If $k\in \mathcal{B}$, we define $c_k$
as,
\begin{align}
c_k=\frac{\sum_{n=1}^{N}m_{n1}I\left(x_n\in
\left[(k-1/2)W_N,(k+1/2)W_N\right]\right)} {
\sum_{n=1}^{N}I\left(x_n\in
\left[(k-1/2)W_N,(k+1/2)W_N\right]\right)},
\end{align}
where $I(\cdot)$ is the indicator function.

Let $\Theta$ denote the 6-tuple
$\{\alpha_1,\alpha_2,\mu_1,\mu_2,\sigma_1,\sigma_2\}$. We use
$P(x;\Theta)$ to denote the mixture Gaussian distribution,
\begin{align} \left(\sum_{i=1}^{2}
\frac{\alpha_i}{\sqrt{2\pi}\sigma}_i\exp\left(\frac{-(x-\mu_i)^2}{2\sigma_i^2}\right)
\right).
\end{align}
We use $P_N(k;\Theta)$ to denote the following discrete probability
distribution over the same alphabet of $(\ldots,a_k/N,\ldots)$,
\begin{align}
& P_N(k;\Theta)= \notag \\ & c_P \int_{(k-1/2)W_N}^{(k+1/2)W_N}
\left(\sum_{i=1}^{2}
\frac{\alpha_i}{\sqrt{2\pi}\sigma_i}\exp\left(\frac{-(x-\mu_i)^2}{2\sigma_i^2}\right)
\right)dx,
\end{align}
where $c_P$ is a normalization constant, $c_P\rightarrow 1$, as
$N\rightarrow \infty$.

\begin{lemma}
\label{outlier_lemma}
\begin{align}
\mathbb{P}(\mathcal{O}_N)\leq & \sum_{i=1}^{2} \exp\left( -
\frac{\left(c(N)^{(\frac{1}{2}+\zeta})+\mu_i^\ast\right)^2}{2(\sigma_i^\ast)^2}
+\ln N \right) \notag \\
& + \sum_{i=1}^{2} \exp\left( -
\frac{\left(c(N)^{(\frac{1}{2}+\zeta})-\mu_i^\ast\right)^2}{2(\sigma_i^\ast)^2}
+\ln N \right)
\end{align}
\end{lemma}
\begin{proof}
\begin{align}
& \mathbb{P}(\mathcal{O}_N)\stackrel{(a)}{\leq}
\sum_{n=1}^{N}\mathbb{P}(|X_n|\geq
M_N) \notag \\
& = \sum_{n=1}^{N}\sum_{i=1}^{2}\mathbb{P}(|X_n|\geq M_N|x_n\mbox{is
of class } i)\mathbb{P}(x_n\mbox{ is of class }i) \notag \\
& \leq \sum_{n=1}^{N}\sum_{i=1}^{2}\mathbb{P}(|X_n|\geq
M_N|x_n\mbox{is
of class } i) \notag \\
& \leq N \sum_{i=1}^{2}
Q\left(\frac{M_N+\mu_i^\ast}{\sigma_i^\ast}\right) + N
\sum_{i=1}^{2} Q\left(\frac{M_N-\mu_i^\ast}{\sigma_i^\ast}\right)
\notag \\
& \stackrel{(b)}{\leq} \sum_{i=1}^{2} \exp\left( -
\frac{\left(c(N)^{(\frac{1}{2}+\zeta})+\mu_i^\ast\right)^2}{2(\sigma_i^\ast)^2}
+\ln N \right), \notag \\
& + \sum_{i=1}^{2} \exp\left( -
\frac{\left(c(N)^{(\frac{1}{2}+\zeta})-\mu_i^\ast\right)^2}{2(\sigma_i^\ast)^2}
+\ln N \right),
\end{align}
where, $Q(\cdot)$ denotes the well known Gaussian tail function, (a)
follows from the union bound, and (b) follows from the well known
Chernoff bound $Q(x)\leq \exp(-x^2/2)$ (see for example
\cite[Section 2-1-5]{proakis}).
\end{proof}

\begin{lemma}
\label{type_class_bound} Let
$\widehat{\Theta}=\{\widehat{\alpha}_1,\widehat{\alpha}_2,\widehat{\mu}_1,\widehat{\mu}_2,\widehat{\sigma}_1,\widehat{\sigma}_2\}$
be the estimated model parameters. Assume that the random event
$\mathcal{O}_N$ does not occur. Then, the following bound holds,
which relates the type $\{a_k/N\}$ to the estimated probability
model parameters.
\begin{align}
 &
D\left(\frac{a_k}{N}||P_N(k;\widehat{\Theta})\right)+H\left(\frac{a_k}{N}\right)
\notag \\
& \leq   \frac{\widehat{\alpha}_1}{2}\ln(2\pi e
\widehat{\sigma}_1^2) + \frac{\widehat{\alpha}_2}{2}\ln(2\pi e
\widehat{\sigma}_2^2) + H(\widehat{\alpha}_1,\widehat{\alpha}_2)
\notag \\ & +
\left(\frac{\widehat{\alpha}_1}{2\widehat{\sigma}_1^2}+\frac{\widehat{\alpha}_2}{2\widehat{\sigma}_2^2}\right)W_N^2
 - \ln W_N -\ln c_P.
\end{align}
\end{lemma}
\begin{proof}\emph{(sketch)}
The  bound is proved in Eq. \ref{long_bound_one}, where (a) follows
from the fact that $\ln(\cdot)$ is an increasing function, (b)
follows from the mean-value theorem, and (c) follows from Eqs.
\ref{optimal_cond_9}, \ref{optimal_cond_10}.

\begin{figure*}
\begin{align}
\label{long_bound_one} &
D\left(a_k/N||P_N(k;\widehat{\Theta})\right)+H\left(a_k/N\right)
 = \sum_k \frac{-a_k}{N}\ln\left(P_N(k;\widehat{\Theta})\right)
 \notag \\
 & = \sum_k
 \frac{-a_k}{N}\ln\left[\int_{(k-1/2)W_N}^{(k+1/2)W_N}
\frac{\widehat{\alpha}_1}{\sqrt{2\pi}\widehat{\sigma}_1}\exp\left(\frac{-(x-\widehat{\mu}_1)^2}{2\widehat{\sigma}_1^2}\right)+
\frac{\widehat{\alpha}_2}{\sqrt{2\pi}\widehat{\sigma}_2}\exp\left(\frac{-(x-\widehat{\mu}_2)^2}{2\widehat{\sigma}_2^2}\right)
dx\right]  -\ln(c_P) \notag \\
& \stackrel{(a)}{\leq } \sum_{k\in {\mathcal A}}
\frac{-a_k}{N}\ln\left[\int_{(k-1/2)W_N}^{(k+1/2)W_N}
\frac{\widehat{\alpha}_1}{\sqrt{2\pi}\widehat{\sigma}_1}\exp\left(\frac{-(x-\widehat{\mu}_1)^2}{2\widehat{\sigma}_1^2}\right)dx\right]
+ \sum_{k\in {\mathcal B}}
\frac{-c_ka_k}{N}\ln\left[\int_{(k-1/2)W_N}^{(k+1/2)W_N}
\frac{\widehat{\alpha}_1}{\sqrt{2\pi}\widehat{\sigma}_1}\exp\left(\frac{-(x-\widehat{\mu}_1)^2}{2\widehat{\sigma}_1^2}\right)dx\right]
\notag \\
 &
 + \sum_{k\in {\mathcal B}}
\frac{-(1-c_k)a_k}{N}\ln\left[\int_{(k-1/2)W_N}^{(k+1/2)W_N}
\frac{\widehat{\alpha}_2}{\sqrt{2\pi}\widehat{\sigma}_2}\exp\left(\frac{-(x-\widehat{\mu}_2)^2}{2\widehat{\sigma}_2^2}\right)dx\right]
\notag \\ &   + \sum_{k\in {\mathcal C}}
\frac{-a_k}{N}\ln\left[\int_{(k-1/2)W_N}^{(k+1/2)W_N}
\frac{\widehat{\alpha}_2}{\sqrt{2\pi}\widehat{\sigma}_2}\exp\left(\frac{-(x-\widehat{\mu}_2)^2}{2\widehat{\sigma}_2^2}\right)
dx\right]
-\ln(c_P)   \notag \\
& = \sum_{k\in {\mathcal A}}
\frac{-a_k}{N}\ln\left[\int_{(k-1/2)W_N}^{(k+1/2)W_N}
\exp\left(\frac{-(x-\widehat{\mu}_1)^2}{2\widehat{\sigma}_1^2}\right)dx\right]
+ \sum_{k\in {\mathcal B}}
\frac{-c_ka_k}{N}\ln\left[\int_{(k-1/2)W_N}^{(k+1/2)W_N}
\exp\left(\frac{-(x-\widehat{\mu}_1)^2}{2\widehat{\sigma}_1^2}\right)dx\right]
\notag \\
 &
 + \sum_{k\in {\mathcal B}}
\frac{-(1-c_k)a_k}{N}\ln\left[\int_{(k-1/2)W_N}^{(k+1/2)W_N}
\exp\left(\frac{-(x-\widehat{\mu}_2)^2}{2\widehat{\sigma}_2^2}\right)dx\right]
+ \sum_{k\in {\mathcal C}}
\frac{-a_k}{N}\ln\left[\int_{(k-1/2)W_N}^{(k+1/2)W_N}
\exp\left(\frac{-(x-\widehat{\mu}_2)^2}{2\widehat{\sigma}_2^2}\right)
dx\right]   \notag \\
& \hspace{0.5in}+  \frac{\widehat{\alpha}_1}{2}\ln(2\pi
\widehat{\sigma}_1^2) + \frac{\widehat{\alpha}_2}{2}\ln(2\pi
\widehat{\sigma}_2^2) + H(\widehat{\alpha}_1,\widehat{\alpha}_2)  - \ln(c_P) \notag \\
 & \stackrel{(b)}{\leq}
\frac{\sum_{n\in {\mathcal A}} (x_n-\widehat{\mu}_1)^2+
\sum_{n\in{\mathcal B}}\widehat{m}_{n1}(x_n-\widehat{\mu}_1)^2
}{2\widehat{\sigma}_1^2  N} + \frac{\sum_{n\in {\mathcal C}}
(x_n-\widehat{\mu}_2)^2+ \sum_{n\in{\mathcal
B}}\widehat{m}_{n2}(x_n-\widehat{\mu}_2)^2 }{2\widehat{\sigma}_2^2
N} \notag \\ & + \frac{\sum_{n\in {\mathcal A}}
(2x_n-2\widehat{\mu}_1+W_N)W_N+ \sum_{n\in{\mathcal
B}}\widehat{m}_{n1}(2x_n-2\widehat{\mu}_1+W_N)W_N
}{2\widehat{\sigma}_1^2  N} \notag \\ & + \frac{\sum_{n\in {\mathcal
C}} (2x_n-2\widehat{\mu}_2+W_N)W_N  + \sum_{n\in{\mathcal
B}}\widehat{m}_{n2}(2x_n-2\widehat{\mu}_2+W_N)W_N
}{2\widehat{\sigma}_2^2 N} \notag \\ & \hspace{0.5in}+
\frac{\widehat{\alpha}_1}{2}\ln(2\pi
 \widehat{\sigma}_1^2) + \frac{\widehat{\alpha}_2}{2}\ln(2\pi
\widehat{\sigma}_2^2) + H(\widehat{\alpha}_1,\widehat{\alpha}_2)
-\ln(W_N) - \ln(c_P) \notag \\
 & \stackrel{(c)}{=}
  \frac{\widehat{\alpha}_1}{2}\ln(2\pi e \widehat{\sigma}_1^2) +
\frac{\widehat{\alpha}_2}{2}\ln(2\pi e \widehat{\sigma}_2^2) +
H(\widehat{\alpha}_1,\widehat{\alpha}_2)  +
\left(\frac{\widehat{\alpha}_1}{2\widehat{\sigma}_1^2}+\frac{\widehat{\alpha}_2}{2\widehat{\sigma}_2^2}\right)W_N^2
 -\ln (W_N) - \ln(c_P)
\end{align}
\line(1,0){500}
\end{figure*}
\end{proof}

\begin{theorem}
\label{error_exp_main_theorem} Let $\mathcal D$ denote a set of
mixture Gaussian distributions with parameters
$\{\alpha_1,\alpha_2,\mu_1,\mu_2,\sigma_1,\sigma_2\}$, where
$\sigma_i^2$ are lower bounded by a positive constant $B_{\sigma}$.
Assume that the true model distribution
$\Theta^\ast\notin\mathcal{D}$, $\Theta^\ast =
\{\alpha_1^\ast,\alpha_2^\ast,\mu_1^\ast,\mu_2^\ast,\sigma_1^\ast,\sigma_2^\ast\}$.
Define the error exponent $E_r({\mathcal D})= \lim_{N\rightarrow
\infty} \frac{-1}{N}\ln{\mathbb P}\left(\widehat{P}\in{\mathcal
D}\right)$. Let $\mathcal{F}(\mathcal{D})$ denote the set of
probability distribution with well-defined probability density
function $P(x)$, such that, there exists a
$P(x;\widetilde{\Theta})$, $\widetilde{\Theta}\in \mathcal{D}$,
$\widetilde{\Theta}=\{\widetilde{\alpha}_1,\widetilde{\alpha}_2,\widetilde{\mu}_1,
\widetilde{\mu}_2, \widetilde{\sigma}_1, \widetilde{\sigma}_2\}$ and
\begin{align}
& D(P(x)||P(x;\widetilde{\Theta}))+H(P(x)) \notag \\
& \leq \frac{\widetilde{\alpha}_1}{2}\ln(2\pi e
\widetilde{\sigma}_1^2) + \frac{\widetilde{\alpha}_2}{2}\ln(2\pi e
\widetilde{\sigma}_2^2) +
H(\widetilde{\alpha}_1,\widetilde{\alpha}_2).
\end{align}
Then $E_r({\mathcal D})\geq E_b$, where
\begin{align}
 E_b({\mathcal D}) = \min_{P\in \mathcal{F}(\mathcal{D})}
D(P(x)||P(x;\Theta^\ast))
\end{align}
\end{theorem}
\begin{proof}\emph{(sketch)}
According to Lemma \ref{outlier_lemma}, the exponent of the random
event $\mathbb{P}(\mathcal{O}_N)$ is infinity. Therefore,
\begin{align}
\mathbb{P}(\widehat{\Theta}\in \mathcal{D}) & =
\mathbb{P}(\widehat{\Theta}\in \mathcal{D},\mathcal{O}_N^c)+
\mathbb{P}(\widehat{\Theta}\in
\mathcal{D}|\mathcal{O}_N)\mathbb{P}(\mathcal{O}_N) \notag \\
& \sim \mathbb{P}(\widehat{\Theta}\in \mathcal{D},\mathcal{O}_N^c).
\end{align}

Let $\mathcal{F}(\mathcal{D},\epsilon)$ denote the set of
probability distribution with probability density function $P(x)$,
such that, there exists $P(x;\widetilde{\Theta})$,
$\widetilde{\Theta}\in \mathcal{D}$,
$\widetilde{\Theta}=\{\widetilde{\alpha}_1,\widetilde{\alpha}_2,\widetilde{\mu}_1,
\widetilde{\mu}_2, \widetilde{\sigma}_1, \widetilde{\sigma}_2\}$,
and
\begin{align}
& D(P(x)||P(x;\widetilde{\Theta}))+H(P(x)) \notag \\
& \leq \frac{\widetilde{\alpha}_1}{2}\ln(2\pi e
\widetilde{\sigma}_1^2) + \frac{\widetilde{\alpha}_2}{2}\ln(2\pi e
\widetilde{\sigma}_2^2) +
H(\widetilde{\alpha}_1,\widetilde{\alpha}_2) +\epsilon.
\end{align}

Let $\mathcal{G}(\mathcal{D},N,\epsilon)$ denote the set of type $P$
of sequences with length $N$, such that, there exists
$P(x;\widetilde{\Theta})$, $\widetilde{\Theta}\in \mathcal{D}$,
$\widetilde{\Theta}=\{\widetilde{\alpha}_1,\widetilde{\alpha}_2,\widetilde{\mu}_1,
\widetilde{\mu}_2, \widetilde{\sigma}_1, \widetilde{\sigma}_2\}$,
and
\begin{align}
& D(P(x)||P_N(k;\widetilde{\Theta}))+H(P(x)) \notag \\
& \leq \sum_{i=1}^{2}\frac{\widetilde{\alpha}_i}{2}\ln(2\pi e
\widetilde{\sigma}_i^2) +
H(\widetilde{\alpha}_1,\widetilde{\alpha}_2) -\ln W_N+\epsilon.
\end{align}

According to Lemma \ref{type_class_bound}, if $\widehat{\Theta}\in
\mathcal{D}$, and $\mathcal{O}_N$ does not occur, then the type
$\{a_k/N\}\in \mathcal{G}(\mathcal{D},N,\epsilon)$. Therefore,
\begin{align}
\mathbb{P}(\widehat{\Theta}\in \mathcal{D},\mathcal{O}_N^c) & \leq
\sum_{P\in \mathcal{G}(\mathcal{D},N,\epsilon)}
\mathbb{P}(\mathcal{T}(P,N)) \notag \\
& \stackrel{(a)}{\sim} \max_{P\in
\mathcal{G}(\mathcal{D},N,\epsilon)} \mathbb{P}(\mathcal{T}(P,N)) \notag \\
& \stackrel{(b)}{\sim} \exp\left(- \min_{P\in
\mathcal{G}(\mathcal{D},N,\epsilon)}D(P||P_N(k;\Theta^\ast))N
\right),
\end{align}
where, (a) follows from the fact that the number of type class is
upper bounded by
\begin{align}
(N+1)^{L_N}\sim 1,
\end{align}
and (b) follows from first principles in the method of types
\cite{csiszar98}.

Let $\{b_k/N\}$ denote the above type, which minimizes
$D(P||P_N(k;\Theta^\ast))$. We can construct a probability
distribution with probability density function $Q(x;N,b_k/N)$ as
follows.
\begin{align}
Q(x;N,b_k/N)=\left\{
\begin{array}{ll}
\frac{b_k}{W_N N}, & \mbox{if }|x|\leq M_N, |x-kW_N|<\frac{W_N}{2}
\\
0, & \mbox{otherwise}
\end{array}
\right.
\end{align}
It can be checked that $Q(x;N,b_k/N)\in
\mathcal{F}(\mathcal{D},\epsilon_1)$, and
\begin{align}
D(b_k/N||P_N(k;\Theta^\ast))\geq D(Q(x;N,b_k/N)||P(x;\Theta^\ast)) -
\epsilon_2,
\end{align}
where $\epsilon_1,\epsilon_2$ are some small positive numbers,
$\epsilon_1, \epsilon_2\rightarrow 0$, as $N\rightarrow \infty$.

As a consequence,
\begin{align}
& \min_{P\in
\mathcal{G}(\mathcal{D},N,\epsilon)}D(P||P_N(k;\Theta^\ast)) \notag
\\
&\geq \min_{P\in \mathcal{F}(\mathcal{D},\epsilon_1)}
D(P(x)||P(x;\Theta^\ast))-\epsilon_2.
\end{align}
Finally, the theorem follows from the fact that all information
divergence and entropy functions are continuous.
\end{proof}

\begin{theorem}
\label{prob_one_theorem} For sufficiently large $N$, with
probability close to one,
\begin{align}
D(P(x;\Theta^\ast)||P(x;\widehat{\Theta}))\leq
H(\alpha_1^\ast,\alpha_2^\ast)+\epsilon,
\end{align}
where $\epsilon$ is a small positive number, $\epsilon\rightarrow
0$, as $N\rightarrow \infty$.
\end{theorem}
\begin{proof} \emph{(sketch)}
We define $a_k/N$ and $Q(x;N,a_k/N)$ similarly as in the above. With
probability close to one, $\mathcal{O}_N$ does not occur, and
\begin{align}
\left| D(P(x;\Theta^\ast)||P(x;\widehat{\Theta}))-
D(Q(x;N,a_k/N)||P(x;\widehat{\Theta}))\right| \leq \epsilon_1.
\end{align}
Note that
\begin{align}
H(P(x;\Theta^\ast))\geq
 \frac{\alpha_1^\ast}{2}\ln\left(2\pi e (\sigma_1^\ast)^2\right) +
\frac{\alpha_2^\ast}{2}\ln\left(2\pi e (\sigma_2^\ast)^2\right).
\end{align}
By Lemma \ref{outlier_lemma} and \ref{type_class_bound}, we have
with probability close to one
\begin{align}
& D(P(x;\Theta^\ast)||P(x;\widehat{\Theta})) \notag \\
& \leq
 \frac{\widehat{\alpha}_1}{2}\ln\left( (\widehat{\sigma}_1)^2\right) +
\frac{\widehat{\alpha}_2}{2}\ln\left( (\widehat{\sigma}_2)^2\right)
+
H(\widehat{\alpha}_1,\widehat{\alpha}_2) \notag \\
& -\frac{\alpha_1^\ast}{2}\ln\left( (\sigma_1^\ast)^2\right) -
\frac{\alpha_2^\ast}{2}\ln\left( (\sigma_2^\ast)^2\right) +
 \epsilon_2.
\end{align}
The theorem then follows from the fact that
\begin{align}
 \frac{\widehat{\alpha}_1}{2}\ln\left( (\widehat{\sigma}_1)^2\right) +
\frac{\widehat{\alpha}_2}{2}\ln\left( (\widehat{\sigma}_2)^2\right)
+ H(\widehat{\alpha}_1,\widehat{\alpha}_2) = \min \frac{\ln(G)}{2};
\end{align}
with probability close to one
\begin{align}
& \frac{\widehat{\alpha}_1}{2}\ln\left(
(\widehat{\sigma}_1)^2\right) +
\frac{\widehat{\alpha}_2}{2}\ln\left( (\widehat{\sigma}_2)^2\right)
+ H(\widehat{\alpha}_1,\widehat{\alpha}_2) \notag \\
& \leq
 \frac{\alpha_1^\ast}{2}\ln\left( (\sigma_1^\ast)^2\right) +
\frac{\alpha_2^\ast}{2}\ln\left( (\sigma_2^\ast)^2\right) +
H(\alpha_1^\ast,\alpha_2^\ast)+
 \epsilon_3.
\end{align}
In the above, $\epsilon_1,\epsilon_2,\epsilon_3$ are all small
positive real numbers, $\epsilon_1,\epsilon_2,\epsilon_3\rightarrow
0$, as $N\rightarrow \infty$.
\end{proof}

\begin{remark}
Theorem \ref{error_exp_main_theorem}, and \ref{prob_one_theorem}
show that the estimated model parameters converge to the true model
parameters in probability except some bias terms. The bound in
Theorem \ref{error_exp_main_theorem} can be further improved.
However, because the discussion is much involved, we leave it for
future research.
\end{remark}

%%%%%%%%%%%%%%%%%%%%%%%%%%%%%%%%%%%%%%%%%%%%%%%%%%%%%%%%%%%%%%%%%%%%%%%%%%%%%%%%%%%%%%%%%%

%

\section{Conclusion}

\label{sec_conclusion}

In this paper, we present an information theoretic performance
analysis of the blind signal classification algorithm proposed in
\cite{ma09}. We show that the obtained classification results in the
algorithm is equivalent to a MAP estimator using the estimated
parametric probability models. We further show that the by-product
model parameter estimation is accurate. These theoretical analysis
suggests that the algorithm has nice performance.

%\nocite{*}
\bibliographystyle{IEEEtran}
\bibliography{isit_2010}

\end{document}